\documentclass[aps,pra,twocolumn,superscriptaddress,longbibliography]{revtex4-2}
\usepackage{amsmath}
\usepackage{amssymb}
\usepackage{soul}
\usepackage{amsthm}
\usepackage{xparse,physics}
\usepackage{tikz}
\usepackage{lipsum}

\newtheorem{thm}{Theorem}
\newtheorem{lem}{Lemma}

\newcommand{\paren}[1]{\qty(#1)}
\newcommand{\cH}{\mathcal{H}}

\begin{document}
\renewcommand{\selectlanguage}[1]{}

\title{A quantum algorithm for modular flow}

\author{Ian T. Lim}
\email{ilim@arizona.edu}
\affiliation{Department of Physics, University of Arizona, Tucson, AZ 85721, USA}
\author{Isaac H. Kim}
\affiliation{Department of Computer Science, UC Davis, Davis, CA 95616, USA}

\begin{abstract}
    Entanglement is a defining property of quantum systems. For a subsystem of a larger quantum system, we can formally define an operator known as the modular Hamiltonian, which is closely linked to the entanglement properties of that subsystem, and a corresponding operator flow called the modular flow. Algorithms for estimating the von Neumann entropy, the best-known entanglement measure, are well-established, but no equivalent procedures have been previously described for the modular flow. In this work, we briefly review the quantum singular value transform (QSVT) framework for developing quantum algorithms, and then discuss the implementation of modular flow within this framework. We conclude by describing select applications of our modular flow algorithm, such as extracting the chiral central charge of a topologically ordered system and simulating the experience of the bulk observer in holography. We also prove a query complexity lower bound for modular flow, which shows that our method cannot be improved further substantially.
\end{abstract}
\maketitle

\section{Introduction}

Over the past few decades, studies of entanglement in quantum many-body systems
have proven to be a fruitful approach to uncover their universal properties. In condensed matter physics, entanglement provided smoking-gun signatures of topological phases 
~\cite{kitaev2006topological,levin2006detecting}. In the context of the AdS/CFT correspondence~\cite{maldacena1999large}, the Ryu-Takayanagi formula revealed a geometric interpretation of entanglement~\cite{ryu_holographic_2006}. Entanglement has also played an important role in elucidating fundamental structural properties of quantum many-body systems, such as the monotonicity of the renormalization group flow~\cite{casini2004finite,casini2012renormalization} and topological phases~\cite{shi2020fusion,kim2023universal}.

The central object that underpins all these studies is the von Neumann entropy. However, von Neumann entropy is just a single number, and the underlying density matrix may contain additional information not revealed by the entropy alone. One way to recover such information is to use the \emph{modular flow}. Modular flow is a time evolution generated by the modular Hamiltonian, which is defined as the logarithm of a given density matrix. Interestingly, modular flow sometimes reveal extra information that is not encoded in the entanglement measures such as von Neumann entropy. For instance, modular flow can reveal the chiral central charge of the edge associated with a topologically ordered 2+1D medium~\cite{kim2022chiral,kim_modular_2022,li2024strictarealawentanglement,fan2022entanglement,zou2022modular}. Modular flow was also studied in the context of the AdS/CFT correspondence~\cite{faulkner_bulk_2018,jafferis_relative_2015}. In particular, it was recently proposed that the modular flow on the CFT side can generate a time evolution of an observer embedded in the AdS~\cite{jafferis_inside_2021,de_boer_black_2022}.

An ability to experimentally implement modular flow may be useful for these applications. For instance, similar to how entanglement-based probe was used to detect topological phases in experiments~\cite{satzinger2021realizing,verresen2021prediction}, generating a modular flow can reveal other universal properties of the underlying phase such as the chiral central charge. In the AdS/CFT correspondence, modular flow gives a testable prediction on the experience of the bulk observer~\cite{jafferis_inside_2021,de_boer_black_2022}. A concrete method to apply the modular flow may lead to an experiment that can verify or disprove this proposal.

However, to the best of our knowledge, there is no known method to implement the modular flow. The main purpose of our work is to propose a method with a provable guarantee, together with the query complexity lower bound that shows the (near-)optimality of our method.  

Our main result is the construction of a polynomial approximation to the logarithm which is suitable for implementation in the quantum singular value transform (QSVT) framework~\cite{gilyen_quantum_2019}. By applying this polynomial to a state $\rho$, we can approximate its modular Hamiltonian $K := -\log\rho$ and use the modular Hamiltonian as input to an established Hamiltonian simulation algorithm. This provides an explicit prescription for computing the modular flow of an operator using QSVT.

The technical result is summarized in the following theorem:
\begin{thm}\label{thm:modularflowcomplexity}
    For a block-encoded state $\rho$ with smallest non-zero eigenvalue $1/\kappa$, we can $\epsilon$-approximate the modular flow $O(t)$ of an operator $O$ with respect to $\rho$ for a modular time $t$, i.e. implement an operator $\tilde O(t)$ such that
    \begin{equation}
        \norm{\tilde O(t) - \rho^{-it}O\rho^{it}} < \epsilon,
    \end{equation}
    with a query complexity of
    \begin{equation}
        \tilde O\paren{ \kappa \abs{t}\log \kappa \log \frac{1}{\epsilon} }
    \end{equation}
    queries to the encoding of $\rho$, where $\tilde O$ indicates the suppression of sub-logarithmic factors.
\end{thm}

We also describe three other relevant extensions. First, we argue that in the context of the Jafferis-Lamprou-Gao proposal relating modular flow for holographic quantum systems to bulk proper time evolution, bulk correlation functions can in principle be evaluated with our modular flow prescription, allowing for experimental verification of the results of Gao and Lamprou \cite{gao_seeing_2022}. Second, we show that our construction for modular flows can be used to estimate the von Neumann entropy of a density matrix. In particular, using the known query complexity lower bound for von Neumann entropy estimation~\cite{chen2025listcomplexityboundsproperty}, we show that the complexity of our algorithm for modular flow cannot be improved substantially. Finally, we show that the modular flow can be used to compute the modular commutator \cite{kim_modular_2022}, an entanglement quantity related to multipartite entanglement on tripartite states.

The remainder of this paper is organized as follows. In Section~\ref{sec:modular_flow_qsvt}, we provide a brief review of QSVT and explain our main algorithm. We discuss applications of this result in Section~\ref{sec:applications} and end with a conclusion in Section~\ref{sec:conclusions}.

\section{Modular flow with QSVT}
\label{sec:modular_flow_qsvt}

\subsection{Review of QSVT}
A key tool we make use of in this work is the quantum singular value transform (QSVT), an algorithmic framework for developing quantum algorithms which allows us to polynomially transform the singular values of a suitably encoded operator \cite{gilyen_quantum_2019,martyn_grand_2021}. The QSVT is based on the principle of \emph{qubitization}, where a bigger Hilbert space is divided into two orthogonal subspaces (analogous to the $\ket{0},\ket{1}$ states of a qubit) and an operator of interest is \emph{block-encoded} to act nontrivially upon the subspace picked out by the $\ket{0}\bra{0}$ projector.

The QSVT generalizes the principle of quantum signal processing for a single qubit \cite{low_optimal_2017,low_hamiltonian_2019}, in which a fixed ``signal'' described by a parameter $a$ and implemented as a black-box operator $W(a)$ (e.g. an $x$-rotation by some fixed angle related to $a$) is ``processed'' by applying a particular sequence of operators $U_{\vec \phi}$ involving $W(a)$ and some ``signal processing'' rotations of our choice (typically taken to be $z$-rotations) to an initial state $\ket{0}$.%
    \footnote{The exact definition is given in \cite{martyn_grand_2021}, for example.}
The signal parameter $a$ is initially unknown to us, but we can apply the operator $W(a)$ at will.

At the end of this process, the resulting state has the special property that the matrix element $\bra{0}U_{\vec \phi}\ket{0}$ is a polynomial function of the parameter $a$. That is, if we prepare a qubit in the $\ket{0}$ state, apply the signal processing sequence to it, and measure in the computational basis, the probability of measuring the state to be $\ket{0}$ at the end is related to the signal parameter $a$ in a simple (polynomial) way.

This procedure would be of limited use if we could not control or easily determine the signal processing polynomial $P(a)$, but in fact for a given polynomial $P(a)$ (subject to some reasonable constraints), we can \emph{always} construct a corresponding signal processing sequence which implements this polynomial so that $P(a)=\bra{0}U_{\vec \phi}\ket{0}$. Thus the ``signal'' $a$ is processed by a polynomial of our choice, and we can for example design this polynomial to let us distinguish whether $a$ is above or below some chosen threshold by simply measuring the output state in the computational basis \cite{gilyen_quantum_2019}.

In the context of QSVT, the polynomial transform of the single parameter $a$ is generalized to a polynomial transform of the singular values of a block-encoded%
    \footnote{Effectively, block encoding the operator $O$ on a Hilbert space $\mathcal{H}_A$ means the following: introduce an index qubit $R$ with Hilbert space $\mathcal{H}_R$. A \emph{block encoding of $O$} is an operator which acts on $\mathcal{H}_R \otimes \mathcal{H}_A$ as $\ket{0}\bra{0}\otimes O$. Block encoding allows us to embed potentially non-unitary operators into a block of a unitary operator acting on a larger Hilbert space.}
operator $A$. In other words, given
\begin{equation}
    A=\sum_{k=1}^r \sigma_k \ket{w_k}\bra{v_k},
\end{equation}
with singular values $\qty{\sigma_k}$ and left- and right-singular vectors $\qty{\ket{w_k}},\qty{\bra{v_k}}$, QSVT implements the transformed version
\begin{equation}
    \mathrm{Poly}(A)=\sum_{k=1}^r P(\sigma_k) \ket{w_k}\bra{v_k}
\end{equation}
This means that if we can reframe a function of an operator (e.g. $e^{iHt}$ for a Hamiltonian $H$ or $A^{-1}$ for an operator $A$) in terms of a polynomial which approximates that function well over the operator's range of eigenvalues, then we can implement that function using QSVT. 

Moreover, up to constant factors, the degree of the polynomial is equal to the number of calls to the block-encoded operator being processed, so constructing the polynomial immediately yields the query complexity of the resulting procedure. QSVT is known to give state-of-the-art asymptotic scaling for various quantum algorithms, and many examples of implementing algorithms with QSVT can be found in \cite{martyn_grand_2021}.

\subsection{Modular flow with approximate log}\label{ssec:modularflowwithlog}

As discussed in the previous subsection, the quantum singular value transform allows us to apply a polynomial of our choice to the singular values of a block-encoded operator. In other words, to implement the operator $K_\rho=-\log \rho$ given a block encoding of $\rho$, we need to construct an appropriate polynomial approximation of the logarithm, subject to a few reasonable conditions.

We begin by invoking a result of Gily\'en and Li \cite{gilyen_distributional_2020}, which provides an appropriate polynomial approximation of the logarithm for use with QSVT. We review their result here before applying Hamiltonian simulation to the output of this subroutine to implement the modular flow.

In Lemma 11 of \cite{gilyen_distributional_2020}, Gily\'en and Li establish that there exists an efficiently computable polynomial $\tilde S$ such that for $\beta\in(0,1],\eta\in  (0,\frac{1}{2}]$, $\forall x\in[\beta,1]$,
\begin{equation}
    \abs{\tilde S(x) - \frac{\ln(1/x)}{2\ln(2/\beta)}} \leq \eta,
\end{equation}
and $\forall x\in[-1,1], -1\leq \tilde S(x)=\tilde S(-x)\leq 1$ ($\tilde S$ is of even parity and bounded between $-1$ and $1$), and
\begin{equation}
    \mathrm{deg}(\tilde S)=\mathcal{O}\qty(\frac{1}{\beta}\log(\frac{1}{\eta})).
\end{equation}

Suppose that our state $\rho$ has condition number $\kappa$ (\emph{i.e.}, smallest eigenvalue $1/\kappa$), and we wish to approximate the logarithm to additive error $\epsilon$. Then we must choose $\beta=1/\kappa$ and $\eta=\frac{\epsilon}{2\ln (2\kappa)}$, and the resulting polynomial is of degree
\begin{equation}
    \mathrm{deg}(\tilde S_{\kappa,\epsilon})=\mathcal{O}\qty(\kappa\log(\frac{\log\kappa}{\epsilon})).
\end{equation}
This gives the complexity of constructing the modular Hamiltonian.

We next consider the overall complexity of implementing the modular flow of an operator, which we do by running a Hamiltonian simulation algorithm on our approximate log of $\rho$. Given an operator $O$ and a desired time $t$, Hamiltonian simulation prepares the operator $e^{iOt}$ to within some error tolerance $\epsilon$. Applying this routine to $\log \rho$ gives $e^{i t \log \rho} = \rho^{it}$, which is what we need in order to implement the modular flow.

Performing Hamiltonian simulation is a straightforward application of the QSVT framework, where the functions to be applied are polynomial approximations to the sine and cosine. The runtime (query complexity) of Hamiltonian simulation naturally depends on the (modular) time $t$ and the error tolerance $\epsilon$. As discussed in \cite{martyn_grand_2021}, given a Hamiltonian $\cH$, a desired time $t$, an error tolerance $\epsilon$, and an $\alpha \geq \norm{\cH}$, Hamiltonian simulation by QSVT can be implemented with
\begin{widetext}
\begin{equation}
    \Theta \paren{\alpha\abs{t} + \frac{\log(1/\epsilon)}{\log (e+ \frac{\log(1/\epsilon)}{\alpha \abs{t}})}}
    \text{ queries to (the encoding of) }\cH/\alpha.
\end{equation}
\end{widetext}
In our case, since the output of the log subroutine is a rescaled version of $\log \rho$ such that $\alpha=1$, we will need to multiply the time $t$ by this scaling factor, $2\log(2\kappa)$. This suggests that the final runtime will be logarithmic in the condition number $\kappa$, i.e. it runs in
\begin{widetext}
\begin{equation}
    \Theta \paren{\abs{t}\log \kappa + \frac{\log(1/\epsilon)}{\log (e+ \frac{\log(1/\epsilon)}{\abs{t}\log \kappa})}}
    \text{ queries to }\cH,
\end{equation}
\end{widetext}
and in turn $\cH$ is prepared by $O(N)$ queries to the block encoding of $\rho$, i.e.
\begin{equation}
    O\paren{\kappa \log \frac{\log \kappa}{\epsilon}}\text{ queries to the encoding of }\rho.
\end{equation}
It follows that the total number of queries to the encoding of $\rho$ is the product of these two expressions,
\begin{widetext}
\begin{equation}\label{eq:flowquerycomplexity}
    O\paren{\paren{\kappa \log \frac{\log \kappa}{\epsilon}}\paren{\abs{t}\log \kappa + \frac{\log(1/\epsilon)}{\log (e+ \frac{\log(1/\epsilon)}{\abs{t}\log \kappa})}}} \text{ queries to the encoding of }\rho.
\end{equation}
\end{widetext}

Suppressing logarithmic factors, we arrive at the result of Thm.~\ref{thm:modularflowcomplexity}:
\begin{equation}
    \tilde O\paren{ \kappa \abs{t}\log \kappa \log \frac{1}{\epsilon} }
\end{equation}
queries to $\rho$ are required to implement modular flow for modular time $t$ for a state of condition number $\kappa$ to precision $\epsilon$.

\subsection{Modular flow applied to the purification}
\label{subsec:purification}

In many applications of modular flow, one often applies it to a state that purifies the given density matrix. In this case, the $\kappa$ dependence can be removed. We discuss such a bound here. 

Without loss of generality, consider a bipartite pure state $|\psi\rangle$ and denote its reduced density matrix over $A$ as $\rho_A$. We describe an algorithm that prepares $(\rho_A^{it}\otimes I_B) |\psi\rangle$ given access to an oracle that prepares $|\psi\rangle$. Note that by the construction of Ref.~\cite{low_hamiltonian_2019}, an oracle preparing $|\psi\rangle$ can yield a block encoding of $\rho_A$. 

The approach is to simply apply our main result [Theorem~\ref{thm:modularflowcomplexity}] to $|\psi\rangle$ with a judicious choice of $\kappa$. This prepares a state $(\tilde{\rho}_A^{it}\otimes I)|\psi\rangle$, where $\tilde{\rho}_A$ is a non-negative matrix whose eigenvectors are the same as $\rho_A$, but with different eigenvalues. More precisely, consider the spectral decomposition of these matrices:
\begin{equation}
\begin{aligned}
    \rho_A &= \sum_i \lambda_i |i{\rangle\langle} i|, \\
    \tilde{\rho}_A &= \sum_i \tilde{\lambda}_i |i{\rangle \langle} i|.
\end{aligned}
\end{equation}
These eigenvalues satisfy $|\lambda_i - \tilde{\lambda}_i|\leq \epsilon$ if $\lambda_i \geq \kappa^{-1}$.

Therefore, we can estimate the difference between $(\tilde{\rho}_A^{it}\otimes I)|\psi\rangle$ and $(\rho_A^{it}\otimes I)|\psi\rangle$ as
\begin{widetext}
\begin{equation}
\begin{aligned}
    \|(\tilde{\rho}_A^{it}\otimes I)|\psi\rangle - (\rho_A^{it}\otimes I)|\psi\rangle\|&\leq \left(\max_{j: \lambda_j\geq \kappa^{-1}} |\lambda_j^{it} - \tilde{\lambda}_j^{it}|\right) + 2\left(\sum_{j:\lambda_j< \kappa^{-1}} \lambda_j^{\frac{1}{2}} \right) \\
    &\leq |t|\epsilon \kappa  + 2\kappa^{-\frac{1}{2}} d.
\end{aligned}
\end{equation}
\end{widetext}
where $d$ is the dimension of $A$. By choosing $\kappa = \left(\frac{d}{|t|\epsilon}\right)^{2/3}$, we get an upper bound of $3d^{2/3} (|t|\epsilon)^{1/3}$. In order to ensure this is at most $\delta$, we can set $\epsilon = O(\delta^3/(d^2|t|))$, yielding $\kappa = O(d^2/\delta^2)$. Plugging this into Theorem~\ref{thm:modularflowcomplexity}, we obtain a bound that scales linearly with $|t|$ and polynomially with $d$.

\section{Applications}
\label{sec:applications}
In this section, we discuss a few select applications of the modular flow.

\subsection{Modular flow in condensed matter systems}
In~\cite{li2024strictarealawentanglement}, the authors showed that the modular flow can reveal a universal property of some topological phases. For concreteness, consider a gapped ground state of some model in two spatial dimensions. The authors considered a tripartition of a disk [Fig.~\ref{fig:abc}] and studied the entanglement entropy of $BC$, assuming a modular flow $\rho_{AB}^{it}$ is applied to the underlying state. The main result is that the entanglement entropy satisfies the following formula:
\begin{equation}
    S(\rho_{BC}(t)) = S(\rho_{BC}(0)) + \frac{\pi c_-}{3} t, \label{eq:ccc}
\end{equation}
where $\rho_{BC}(t)$ is the reduced density matrix over $BC$ under the modular flow $\rho_{AB}^{it}$.

The important quantity in Eq.~\eqref{eq:ccc} is $c_-$, which is the chiral central charge~\cite{kim2022chiral,kim_modular_2022}. Importantly, its value can be obtained by computing Eq.~\eqref{eq:ccc} for two different values of $t$, a task that can be achieved by applying our modular flow algorithm and then measuring the von Neumann entropy.

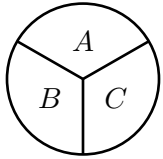
\begin{figure}[h]
    \centering
    \begin{tikzpicture}[line width=1pt]
    \draw[] (0,0) circle (1cm);
    \draw[] (0,0) -- (30:1cm);
    \draw[] (0,0) -- (150:1cm);
    \draw[] (0,0) -- (270:1cm);
    \node[] () at (90:0.5cm) {$A$};
    \node[] () at (210:0.5cm) {$B$};
    \node[] () at (330:0.5cm) {$C$};
        
    \end{tikzpicture}
    \caption{Partition of a disk used in Ref.~\cite{li2024strictarealawentanglement}.}
    \label{fig:abc}
\end{figure}

For this application, it is possible to tame the condition number $\kappa$-dependence of the algorithm; see Section~\ref{subsec:purification} for details.

\subsection{Modular flow in AdS/CFT}
In \cite{jafferis_inside_2021}, Jafferis and Lamprou proposed a framework for ``describing the experience of bulk observers in AdS/CFT.'' Roughly, in the AdS/CFT correspondence, properties of a certain gravitational theory in $d+1$ dimensions (the bulk) are shown to correspond in a precise way to properties of a quantum theory in $d$ dimensions (the boundary) and vice versa. Typically, entanglement quantities on the boundary are found to have geometric duals in the bulk, with the most famous example being the Ryu-Takayanagi formula and its generalizations \cite{ryu_holographic_2006,hubeny_covariant_2007} relating entropy in the boundary to extremal surfaces in the bulk.

The objective of Jafferis and Lamprou was to formulate a boundary description of the proper time evolution of the local fields accessible to a bulk ``observer,'' i.e. a test particle moving in a prepared bulk spacetime geometry with access to fields within a radius $r$ of the particle's position in space.

More precisely, in the setting of \cite{jafferis_inside_2021} a black hole ``probe'' is prepared in the bulk by entangling some subsystem of the boundary CFT with a reference system and then tracing out the reference. In the dual bulk description, this corresponds to introducing a probe black hole near the boundary of AdS. After we trace out the reference system, the CFT is described by a reduced state $\rho$ and a corresponding modular Hamiltonian $K\equiv -\log \rho$.

At some initial time, the observer may have access to some ``atmosphere'' operators nearby in space. As the observer moves in spacetime, the fields she has access to also experience (proper) time evolution along her worldline. The Heisenberg-picture evolution of these atmosphere operators is given by
\begin{equation}
    \phi_H^t(x) = V_H(0,t) \phi^0_H(x) V_H^\dagger(0,t)
\end{equation}
for some boundary operator $V_H(0,t)$.

The key insight of Jafferis and Lamprou is that for certain configurations of the boundary theory, it can be shown that this operator is given by
\begin{equation}
    V_H(0,t) = \exp\bqty{-i \frac{\tau(t)}{2\pi}K(0) + \text{zero modes}},
\end{equation}
where $K$ is the modular Hamiltonian, $\tau(t)$ is proper time as a function of coordinate time, and the zero modes are said to describe the ``precession of the symmetry frame of the observer'' \cite{jafferis_inside_2021}.

In other words, up to this precession effect, the proper time evolution of bulk operators in the Heisenberg picture is implemented by \emph{modular} time evolution with respect to the modular Hamiltonian on the boundary, and so modular time should be identified with proper time in units of the black hole temperature $\beta/2\pi$.

In the follow-up \cite{gao_seeing_2022} by Gao and Lamprou, this construction was made more explicit in the context of the AdS${}_2$/SYK correspondence. In that work, one considers an SYK thermofield double state, dual to an eternal AdS${}_2$ wormhole solution of Jackiw-Teitelboim gravity.

Gao and Lamprou propose to study a correlation function relating an operator $\phi_l$ inserted in the left SYK boundary system at coordinate time $t$ and a probe operator $\phi_r$ inserted in the right system, entangled with a reference system $\cH_\mathrm{ref}$ as previously described. Tracing out the reference yields a reduced state and a corresponding modular Hamiltonian.

Concretely, they compute the anticommutator
\begin{equation}\label{eq:lampros-correlator}
    W(s,t) := \Tr(\rho \qty{\rho^{-is}\psi_r \rho^{is},\psi_l(t)})
\end{equation}
as a function of modular time $s$. The reduced state $\rho$ is a density matrix on the Hilbert space of $\mathrm{SYK}_l \times \mathrm{SYK}_r$ corresponding to tracing out the reference system for a state
\begin{equation}
    \ket{\beta_l,\beta_r; \delta} := \mathcal{Z}^{-1/2} e^{-\frac{\beta_l H_l}{2}}e^{-\frac{\beta_r H_r}{2}} U_{r,\mathrm{ref}}(\delta) \ket{\Phi}_{lr} \ket{v}_\mathrm{ref}
\end{equation}
on $\mathrm{SYK}_l \times \mathrm{SYK}_r \otimes \cH_\mathrm{ref}$. Here, $\ket{v}_\mathrm{ref}$ is the vacuum state of the reference system, while $U_{r,\mathrm{ref}}(\delta)$ both inserts the probe in $\mathrm{SYK}_r$ and also creates the entanglement between the probe and the reference system.

Since the right operator is initially spacelike separated from the left insertion, this correlation function vanishes. However, by performing modular flow (modular time evolution) on the probe operator, one can study how this correlation function changes as the probe evolves in proper time, eventually crossing the bulk light cone of the left operator insertion. In this way, one can probe behind the horizon of the black hole.

While Gao and Lamprou studied this correlator in the context of SYK, our contribution is to show that given an oracle which implements the reduced state of the holographic quantum system with probe, it is possible to implement the modular flow with respect to this reduced state in a reasonable number of calls to the oracle.%
	\footnote{By reasonable we just mean that the scaling in terms of the desired precision and the modular time is not too bad; a precise scaling will appear later.}
In other words, this correlator can be calculated using a quantum computer and used to study the black hole interior or in principle the singularity of the holographic black hole itself \cite{de_boer_black_2022}.

In subsection \ref{ssec:modularflowwithlog}, we showed that we can implement the operator $e^{i t \log \rho} = \rho^{it}$ to a precision $\epsilon$ with a number of queries to $\rho$ that depends on both modular time $t$ and the smallest eigenvalue $1/\kappa$ of the density matrix $\rho$. However, the quantity we wish to compute is not just $\rho^{it}$; it is the correlation function in Eq.~\eqref{eq:lampros-correlator}, reproduced here:
\begin{equation}
    W(s,t) := \Tr(\rho \qty{\rho^{-is}\psi_r \rho^{is},\psi_l(t)}).
\end{equation}

We have assumed oracle access to a block encoding of $\rho$, and in this section we will also assume oracle access to the operator insertions $\psi_r,\psi_l(t)$. Alternately we could assume access to $\psi_l$ and (a block encoding of) the Hamiltonian $H_l$ of the left SYK system, and perform Hamiltonian simulation for a time $t$ with QSVT, just as we did for the modular Hamiltonian.

In any case, access to $\rho, \psi_r,$ and $\psi_l(t)$ ensures that we can implement the operators
\begin{equation}
    \rho \rho^{-is}\psi_r \rho^{is} \psi_l(t), \quad \rho \psi_l(t) \rho^{-is}\psi_r \rho^{is}
\end{equation}
which appear in Eq.~\eqref{eq:lampros-correlator}.

How can we estimate the trace of these operators? We will invoke another fundamental quantum algorithm, \emph{quantum phase estimation} (QPE). In the usual implementation of quantum phase estimation, we are given access to a unitary $U$ and one of its eigenvectors $\ket{k}$, and the output of QPE is an estimate of the phase eigenvalue $e^{i\theta_k}$ corresponding to $\ket{k}$, i.e. $U\ket{k} = e^{i\theta_k}\ket{k}$.

However, there is a reasonable extension of QPE if we are given a mixed state $\rho$ rather than a single eigenvector $\ket{k}$. In this case, QPE yields (to within some precision) a weighted combination of the eigenvalues,
\begin{equation}
    \mathrm{QPE}(\rho,U) = \sum_k e^{i\theta_k} \bra{k} \rho\ket{k} = \Tr(\rho U) \equiv \mathbb{E}_{\rho,k}[e^{i\theta_k}].
\end{equation}
The weights are given by the diagonal elements $\bra{k} \rho\ket{k}$ in the basis of eigenvectors of $U$.

Using QPE, we want to compute
\begin{equation}
    \Tr(\rho \rho^{-is}\psi_r \rho^{is} \psi_l(t)).
\end{equation}
To do so, we observe that while $\psi_r$ and $\psi_l$ are fermionic operators and not themselves unitary, we can perform a Jordan-Wigner transform to map the fermionic operators into linear combinations of a few operators acting on spins, and these operators \emph{are} unitary. Therefore calculating this trace reduces to computing terms of the form
\begin{equation}
    \Tr(\rho \rho^{-is} U_1 \rho^{is} U_2) = \Tr(\rho \tilde U)
\end{equation}
where $\tilde U := \rho^{-is} U_1 \rho^{is} U_2$ is unitary, and thus we can apply QPE to compute these terms.

In summary, if we have oracle access to the implementations of $\rho, \psi_r,$ and $\psi_l(t)$ on spins, then we can apply quantum phase estimation to calculate
\begin{equation}
    W(s,t) := \Tr(\rho \qty{\rho^{-is}\psi_r \rho^{is},\psi_l(t)}).
\end{equation}

\subsection{von Neumann entropy and query complexity lower bound}\label{ssec:entropy_and_complexity}
Computing the von Neumann entropy of an $n$-dimensional density matrix to $\epsilon$-additive precision is known to have a query complexity scaling like $\Omega( (n/\epsilon)^{\frac{1}{2}} + (\log n)/\epsilon)$~\cite{chen2025listcomplexityboundsproperty}.
By reducing the computation of the von Neumann entropy to an application of our modular flow protocol, we can prove a lower bound on the complexity of computing the modular flow.

\begin{widetext}
\begin{thm}
    Using $O\paren{\qty(\frac{(\log \kappa)^2}{\delta \epsilon^2}) \paren{\kappa \log \frac{\log \kappa}{\epsilon}}\paren{\abs{t}\log \kappa + \frac{\log(1/\epsilon)}{\log (e+ \frac{\log(1/\epsilon)}{\abs{t}\log \kappa})}}}$ samples of the state $\rho$, the von Neumann entropy of the state $\rho$ with smallest eigenvalue $1/\kappa$ can be $\epsilon$-approximated with probability at least $1-\delta$.
\end{thm}
\end{widetext}

Recall that with QSVT, we prepared an $\epsilon$-approximation of $-\log \rho/\log \kappa$, where $1/\kappa$ was the smallest non-zero eigenvalue of $\rho$. Applying the Hamiltonian simulation procedure of our choice (e.g. the one described in \cite{martyn_grand_2021}) to the modular Hamiltonian, we can prepare the operator
\begin{equation}\label{eq:qpe_unitary}
    U:=\exp(it\frac{\log \rho}{\log \kappa})
\end{equation}
for a real-valued modular time $t$ of our choosing, and we can do so with a query complexity of
\begin{widetext}
$O\paren{\paren{\kappa \log \frac{\log \kappa}{\epsilon}}\paren{\abs{t}\log \kappa + \frac{\log(1/\epsilon)}{\log (e+ \frac{\log(1/\epsilon)}{\abs{t}\log \kappa})}}}$
\end{widetext}
queries to the encoding of $\rho$, as described in Theorem~\ref{thm:modularflowcomplexity}.

The operator $U$ is unitary by construction, so we can apply quantum phase estimation.%
    \footnote{One may worry that our modular flow construction only implements an approximate version of the operator $U$, and thus we should account for this error in the estimate. However, since our approximate operator $\tilde U$ is $\epsilon$-close in operator norm to the exact unitary $U$, it follows that the output of quantum phase estimation will also be $O(\epsilon)$-close to the exact phase.}
By the definition of $\kappa$, the eigenvalues of $-\log \rho$ lie in the range $[0,\log \kappa]$. Let us therefore choose $t=\pi$ so that all the eigenvalues of $U$ lie on the upper half of the unit circle in the complex plane.

Applying quantum phase estimation to an operator $V=\sum_n e^{i\theta_n}\ket{n}\bra{n}$ for a mixed state $\rho=\sum_n \lambda_n \ket{n}\bra{n}$ yields
\begin{equation}
    \mathrm{QPE}(V)_\rho = \sum_n \theta_n \lambda_n,
\end{equation}
i.e. an appropriately weighted average over the phases $\theta_n$. Thus, if we apply quantum phase estimation to $U$ for the state $\rho = \sum_n \lambda_n \ket{n}\bra{n}$, it follows that the mean output will be
\begin{equation}
    \sum_n \lambda_n \paren{\log \lambda_n \frac{\pi}{\log \kappa}} = \frac{\pi}{\log \kappa} S(\rho).
\end{equation}
Because we have rescaled all the phases $\theta_n$ to be $O(1)$, the variance $\sigma^2$ of the QPE samples is upper-bounded by an $O(1)$ constant and by the central limit theorem scales as $1/N$ with $N$ the number of samples from QPE. By standard tail bounds, our estimate of $\frac{\pi}{\log \kappa} S(\rho)$ will lie within $\epsilon/\log \kappa$ of the true value with probability $1-\delta$ where $\delta \equiv \frac{(\log \kappa)^2}{N\epsilon^2}$. In other words, we require
\begin{equation}
    N = O\qty(\frac{(\log \kappa)^2}{\delta \epsilon^2})
\end{equation}
samples from QPE to compute $S(\rho)$ to within $\epsilon$-additive error with probability $1-\delta$.

In turn, each sample from QPE requires us to prepare the unitary $U$ as defined in Eq.~\eqref{eq:qpe_unitary}, which requires
\begin{widetext}
$O\paren{\paren{\kappa \log \frac{\log \kappa}{\epsilon}}\paren{\pi\log \kappa + \frac{\log(1/\epsilon)}{\log (e+ \frac{\log(1/\epsilon)}{\pi\log \kappa})}}}$
\end{widetext}
calls to our block encoding of $\rho$.

We therefore find that QSVT enables us to $\epsilon$-additively estimate the von Neumann entropy with probability $1-\delta$ in
\begin{widetext}
\begin{equation}\label{eq:vn_estimate}
    O\paren{\qty(\frac{(\log \kappa)^2}{\delta \epsilon^2}) \paren{\kappa \log \frac{\log \kappa}{\epsilon}}\paren{\pi\log \kappa + \frac{\log(1/\epsilon)}{\log (e+ \frac{\log(1/\epsilon)}{\pi\log \kappa})}}}
\end{equation}
\end{widetext}
calls of the state $\rho$ with smallest eigenvalue $1/\kappa$.

Having computed the complexity in terms of $\kappa,\delta,$ and $\epsilon$, we now wish to find sufficient conditions on $\kappa$ and the dimension $n$ so that
\begin{equation}
    \abs{\Tr(\rho \log \rho) - \Tr (\rho\log \tilde \rho)} \leq \epsilon.
\end{equation}
This will allow us to compare the complexity of our procedure to known lower bounds on the von Neumann entropy in terms of the dimension of the state. Such a condition is given by the following lemma, which we prove in Appendix~\ref{app:kappa-epsilon_proof}.
\begin{lem}\label{lem:kappa-epsilon} 
    Consider a state $\rho$ of dimension $n$. Let $f^\mathrm{log}_{\kappa,\epsilon}(x)$ be a function that $\epsilon$-approximates $\log(x)$ on the interval $[-1,1]\setminus [-\kappa,\kappa]$ and that is bounded between $-\log 2\kappa$ and $\log 2\kappa$ on the whole interval $[-1,1]$. For 
    \begin{equation}
        \kappa> \frac{n\log n}{\epsilon},
    \end{equation}
    the estimate of the von Neumann entropy given by $-\Tr(\rho f_{\kappa,\epsilon}^{\mathrm{log}}(\rho))$ satisfies the bound
    \begin{equation}
        \abs{\Tr(\rho \log \rho) - \Tr(\rho f_{\kappa,\epsilon}^{\mathrm{log}}(\rho))} \leq \epsilon.
    \end{equation}
\end{lem}

Our strategy is to apply Lemma~\ref{lem:kappa-epsilon} to the function $\tilde S_{\kappa,\epsilon}$ discussed in Section~\ref{ssec:modularflowwithlog}, which satisfies these assumptions after rescaling by $\log 2\kappa$.

We previously showed that estimating the von Neumann entropy with modular flow requires
\begin{equation}
    N = O\qty(\frac{(\log \kappa)^2}{\delta \epsilon^2})
\end{equation}
calls to the modular flow. Using the result of Lemma~\ref{lem:kappa-epsilon} to choose $\kappa>n\log n/\epsilon$, we find that the number of calls to modular flow can be written independent of $\kappa$ as
\begin{equation}
    N = O\qty(\frac{(\log (n\log n/\epsilon))^2}{\delta \epsilon^2}).
    \label{eq:modular_flow_invocation}
\end{equation}

From Eq.~\eqref{eq:modular_flow_invocation}, we can immediately see that our algorithm for modular flow cannot be improved substantially. Indeed, suppose our algorithm's dependence on $\kappa$ is polylogarithmic. Note that it suffices to choose $\kappa = O(n\log n /\epsilon)$ to estimate the von Neumann entropy up to an additive error of $\epsilon$ (Lemma~\ref{lem:kappa-epsilon}). Using Eq.~\eqref{eq:modular_flow_invocation}, we see that the complexity of estimating the von Neumann entropy would scale polylogarithmically in $n$. However, there is a known query-complexity lower bound  (in the purified access model) that scales as $\Omega( (n/\epsilon)^{\frac{1}{2}} + (\log n)/\epsilon)$~\cite{chen2025listcomplexityboundsproperty}. As such, the poly-logarithmic dependence on $\kappa$ would violate this bound. (In fact, any sub-polynomial depndence will violate the bound as well.)

\section{Conclusions}
\label{sec:conclusions}

In this work, we have proved that the modular flow can be theoretically implemented with a query complexity of $\tilde O\paren{ \kappa \abs{t}\log  \kappa \log \frac{1}{\epsilon} }$, and that this query complexity is optimal up to polynomial factors.

We have presented an explicit prescription for constructing the modular Hamiltonian and implementing modular flow in the QSVT framework, and identified two concrete applications. In the context of AdS/CFT, the modular flow can be used to explore bulk physics in holographic quantum systems according to the Jafferis-Lamprou proposal. Our construction of the modular Hamiltonian allows for an alternate method of numerically determining the von Neumann entropy of a state, and this reduction allows us to prove a lower bound on any modular flow prescription.

In preparing this work, we believe that a third application, computing the modular commutator of a tripartite state as discussed in \cite{kim_modular_2022}, should be possible with our modular flow procedure. However, completing this argument rigorously seems to require a continuity bound on the modular commutator, which we have thus far been unable to prove.

First, exploring holographic quantum systems with the modular flow in the lab may require a large system size around $O(100)$ fermions (as per estimates from \cite{jafferis_traversable_2022}) and high connectivity between the qubits in order to probe non-trivial bulk physics. As of time of writing, this will require further advancements in experimental implementation to achieve.

Second, the modular flow proposal of Jafferis and Lamprou was developed in the context of a highly symmetric spacetime like the AdS Schwarzschild black hole and for a small code subspace (in the langauge of holographic error correction), and therefore may not generalize readily to other spacetimes of more physical interest.%
    \footnote{We thank Netta Engelhardt and Lisa Yang for pointing this out.}
Nevertheless, our protocol shows that for this restricted class of states, one can at least in principle efficiently test properties relating to observers in these highly symmetric spacetimes.

Finally, we have assumed oracle access to a block encoding of the state $\rho$ used in the modular flow $\rho^{it}$ in order to discuss the query complexity of our procedure. Access to such a block encoding is a standard assumption, but in practice, the states of interest may be difficult to prepare, particularly in the holographic case. It would be interesting to understand the hardness of preparing holographic states for which the modular flow can compute some interesting bulk quantity. We defer this to future work.

\section*{Acknowledgement}
IK acknowledges supports from NSF under award number PHY-2337931. Support for this work was provided by an Alfred P. Sloan Research Fellowship.

\appendix

\section{Proof of Lemma~\ref{lem:kappa-epsilon}}\label{app:kappa-epsilon_proof}
\begin{proof}
By assumption, the function $f^\mathrm{log}_{\kappa',\epsilon'}(x)$ obeys the inequality 
\begin{equation}\label{eq:tightlogbound}
    \abs{f^\mathrm{log}_{\kappa',\epsilon'}(x) - \log x} \leq \epsilon'
\end{equation}
on the interval $[-1,1]\setminus [-1/\kappa',1/\kappa']$. We can therefore split the contributions to $\Tr(\rho \log \rho)$ into the eigenvalues less than $1/\kappa'$ and the eigenvalues greater than $1/\kappa'$, and bound each of these contributions separately.

That is,
\begin{widetext}
\begin{align}
    \abs{\Tr(\rho \log \rho) - \Tr(\rho f_{\kappa',\epsilon'}^{\mathrm{log}}(\rho))} &= \abs{\sum_i \lambda_i \log \lambda_i - \sum_i \lambda_i f_{\kappa',\epsilon'}^{\mathrm{log}}(\lambda_i)}\\
        &\leq \sum_i \abs{\lambda_i \log \lambda_i - \lambda_i f_{\kappa',\epsilon'}^{\mathrm{log}}(\lambda_i)}\\
        &= S_\mathrm{small} + S_\mathrm{large}
\end{align}
where
\begin{equation}
    S_\mathrm{small} \equiv \sum_{\lambda_i \leq 1/\kappa'} \abs{\lambda_i \log \lambda_i - \lambda_i f_{\kappa',\epsilon'}^{\mathrm{log}}(\lambda_i)},
    \quad
    S_\mathrm{large} \equiv \sum_{\lambda_i > 1/\kappa'} \abs{\lambda_i \log \lambda_i - \lambda_i f_{\kappa',\epsilon'}^{\mathrm{log}}(\lambda_i)}
\end{equation}
\end{widetext}
and the first inequality comes from applying the triangle inequality. We can directly bound the second term using Eq.~\eqref{eq:tightlogbound}:
\begin{align}
    S_\mathrm{large} &= \sum_{\lambda_i > 1/\kappa'} \lambda_i \abs{ \log \lambda_i - f_{\kappa',\epsilon'}^{\mathrm{log}}(\lambda_i)} \leq \sum_{\lambda_i > 1/\kappa'} \lambda_i \epsilon' \leq \epsilon'.
\end{align}
Meanwhile, the first term can be bounded as follows. By construction, the function $f^\mathrm{log}_{\kappa',\epsilon'}(x)$ is bounded by $\log 2\kappa'$ on the whole interval $[-1,1]$, so
\begin{widetext}
\begin{equation}
    S_\mathrm{small} = \sum_{\lambda_i \leq 1/\kappa'} \lambda_i \abs{ \log \lambda_i - f_{\kappa',\epsilon'}^{\mathrm{log}}(\lambda_i)}
    \leq \sum_{\lambda_i \leq 1/\kappa'} \lambda_i \qty(\log \frac{1}{\lambda_i} + \log 2\kappa').
\end{equation}
\end{widetext}

Denote $\Delta \equiv \sum_{\lambda_i \leq 1/\kappa'} \lambda_i.$ Since each of these $\lambda_i$s is at most $1/\kappa'$ and there are at most $n$ such eigenvalues, we have the bound
\begin{equation}
    \Delta \leq \frac{n}{\kappa'}.
\end{equation}
We notice that
\begin{equation}
    \sum_{\lambda_i \leq 1/\kappa'} \frac{\lambda_i}{\Delta} \log \frac{1}{\lambda_i/\Delta} \leq \log n,
\end{equation}
where $n$ is the dimension (i.e. there are at most $n$ terms in the sum). This follows since $\qty{\lambda_i/\Delta}_{\lambda_i \leq 1/\kappa'}$ now defines a normalized probability distribution of size at most $n$, and so the LHS is just the Shannon entropy of this distribution.

We can rewrite the LHS of this inequality:
\begin{align}
    \sum_{\lambda_i \leq 1/\kappa'} \frac{\lambda_i}{\Delta} \log \frac{1}{\lambda_i/\Delta}
        &= \sum_{\lambda_i \leq 1/\kappa'} \frac{\lambda_i}{\Delta} \paren{\log \frac{1}{\lambda_i} + \log \Delta}\\
        &= \sum_{\lambda_i \leq 1/\kappa'} \frac{\lambda_i}{\Delta} \paren{\log \frac{1}{\lambda_i}} + \log \Delta,
\end{align}
and so
\begin{equation}
    \sum_{\lambda_i \leq 1/\kappa'} \lambda_i \log \frac{1}{\lambda_i} \leq \Delta \log n - \Delta \log \Delta.
\end{equation}
We find that
\begin{align}
    S_\mathrm{small} &\leq \Delta \log n - \Delta \log \Delta + \Delta \log 2\kappa'\\
        &\leq \frac{n}{\kappa'} \log n - \frac{n}{\kappa'} \log \frac{n}{2\kappa'^2},
\end{align}
when $\Delta \leq \frac{n}{\kappa'} \leq 1/e$.
As $\kappa'$ grows large, $1/\kappa'$ (and therefore the error quantified by $S_\mathrm{small}$) becomes small.

We want the total error to be $O(\epsilon)$:
\begin{widetext}
\begin{equation}
    \abs{\Tr(\rho \log \rho) - \Tr(\rho f_{\kappa',\epsilon'}^{\mathrm{log}}(\rho))} \leq S_\mathrm{small} + S_\mathrm{large} \leq \epsilon' + \frac{n}{\kappa'}\log n - \frac{n}{\kappa'} \log \frac{n}{2\kappa'^2}.
\end{equation}
\end{widetext}
If we choose $\epsilon'=\epsilon$ and $\kappa'> \frac{n\log n}{\epsilon}$, then the whole error is $O(\epsilon)$.

\end{proof}

\bibliography{mybibliography}

\end{document}